\documentclass[a4paper,UKenglish]{lipics}

\usepackage{mathtools}
\usepackage{enumerate}
\usepackage{multirow}
\usepackage{tikz}
\usetikzlibrary{positioning,automata}
\usetikzlibrary{decorations.pathreplacing} 

\usepackage{color, colortbl}
\definecolor{Gray}{gray}{0.92}
\definecolor{Grayish}{gray}{0.95}

\usepackage{extarrows}

\newcommand{\Z}{\mathbb{Z}}
\newcommand{\N}{\mathbb{N}}
\newcommand{\T}{\mathcal{T}}

\newcommand{\bv}{\mathbf v}

\newcommand{\bx}{\mathbf x}
\newcommand{\by}{\mathbf y}

\newcommand{\add}{\textsc{Add}}
\newcommand{\move}{\textsc{Move}}
\newcommand{\ch}{\textsc{Check}}

\newcounter{temp}

\newcounter{reach}
\newcounter{degree}

\begin{document}

\title{Undecidability of Two-dimensional Robot Games}

\author[1]{Reino Niskanen}
\author[1]{Igor Potapov}
\author[2]{Julien Reichert}
\authorrunning{R.~Niskanen, I.~Potapov, and J.~Reichert}

\affil[1]{Department of Computer Science, 
University of Liverpool, UK \\  
 \texttt{$\{$r.niskanen,potapov$\}$@liverpool.ac.uk}
	}

\affil[2]{LSV, ENS Cachan, France \\
 \texttt{reichert@crans.org}
	}
	
\maketitle

\begin{abstract}
Robot game is a two-player vector addition game played on the integer lattice $\Z^n$. Both players have sets of vectors and in each turn the vector chosen by a player is added to the current configuration vector of the game. One of the players, called Eve, tries to play the game from the initial configuration to the origin while the other player, Adam, tries to avoid the origin. The problem is to decide whether or not Eve has a winning strategy.
In this paper we prove undecidability of the robot game in dimension two answering the question formulated by Doyen and Rabinovich in 2011 and closing
the gap between undecidable and decidable cases.
\end{abstract}

\section{Introduction}

In the modern world  the reliability of a software code and verification of the correct functionality of complex technological devices
require  the analysis of various interactive processes and open systems, where it is important to
take into account the effects of uncontrollable adversaries, such as environment or malicious users.
Computational games provide a good framework to model interactive processes and
the extensions of classical reachability problems to game schemes, studied in different contexts and settings,
have recently garnered considerable interest \cite{ABO08,AMSS13,BBE10,BJK2010,CD12,FJLS11,Reichert}.

In this paper we study two-player games where the main problem is to decide which of the players wins 
based on a given set of eligible moves, a computational environment and reachability objectives.
Following early results for games on VASS (Vector Addition Systems with States)\footnote{A game is played on a graph with states of player 1 and states of player 2,
with $\N^2$ as the vector space.} \cite{CSL2003,BJK2010}, Doyen and Rabinovich formulated an open problem 
about the simplest version of  games  (\emph{robot games}) for which the decidability was unknown \cite{RobotGames}. 
Robot games are two-player games played by updating a vector of $n$ integer counters. Each of the players, called Adam and Eve, has a finite set of vectors in $\Z^n$. A play starts from a given initial vector $\bx_0\in\Z^n$, and proceeds in rounds. During each round, first Adam adds a vector from his set, followed by Eve doing the same. Eve wins when, after her turn,  the vector is the zero vector. A simple example of the game is illustrated below.

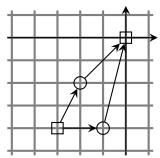
\begin{figure}[htb]
\centering\begin{tikzpicture}[scale=0.6] 
\draw[gray,thick] (0.9,1)--(4.1,1);
\draw[gray,thick] (0.9,0.5)--(4.1,0.5);
\draw[gray,thick] (0.9,0)--(4.1,0);
\draw[gray,thick] (0.9,-0.5)--(4.1,-0.5);
\draw[gray,thick] (0.9,-1)--(4.1,-1);
\draw[gray,thick] (0.9,-1.5)--(4.1,-1.5);
\draw[gray,thick] (0.9,-2)--(4.1,-2);

\draw[gray,thick] (1,-2.1)--(1,1.1);
\draw[gray,thick] (1.5,-2.1)--(1.5,1.1);
\draw[gray,thick] (2,-2.1)--(2,1.1);
\draw[gray,thick] (2.5,-2.1)--(2.5,1.1);
\draw[gray,thick] (3,-2.1)--(3,1.1);
\draw[gray,thick] (3.5,-2.1)--(3.5,1.1);
\draw[gray,thick] (4,-2.1)--(4,1.1);

\draw[-stealth] (0.9,0.5) -- (4.2,0.5);
\draw[-stealth] (3.5,-2.1) -- (3.5,1.2);

\node [rectangle,draw,minimum size=0.3cm,scale=0.5] (A2) at (3.5,0.5) {};
\node [rectangle,draw,minimum size=0.3cm,scale=0.5] (A1) at (2,-1.5) {};
\node [circle,draw,scale=0.5] (E1) at (3,-1.5) {};
\node [circle,draw,scale=0.5] (E2) at (2.5,-0.5) {};
\path[-stealth] (A1) edge (E1)
	(A1) edge (E2)
	(E1) edge (A2)
	(E2) edge (A2);
	\node [align=right,right] at (6,-0.3) {\parbox{4.5cm}{
Adam's moves: $\{(1,2),(2,0)\}$ \\
\phantom{a} \\
Eve's moves: $\{(2,2),(1,4)\}$}};
\end{tikzpicture}
\end{figure}

\noindent We say that Eve has a winning strategy if she eventually can reach the zero vector replying on any Adam move and
Adam has a winning strategy otherwise. Thus a winning strategy gives a way for a player to win, regardless of the way the opponent plays. 
Previously, it has been proved that deciding the winner in one-dimensional robot games is EXPTIME-complete \cite{ArulReichert}.

In this paper we consider the open problem of deciding the winner of robot games for dimension $n=2$ and show that it is undecidable
to check which of the players has a winning strategy in a two-dimensional robot game, i.e., in a very restricted fragment of counter reachability games
with stateless players playing in integer grid $\Z^2$.
The basis of proofs are 2-counter Minsky machines (2CM) for which the halting problem is undecidable. For a 2-counter machine, we construct a game where Eve has to simulate the machine and Adam verifies that Eve does not cheat. The intuition is that the counters of the machine are multiplied by constants and represented by two-dimensional vectors. Additionally, the states of the machine are encoded in the least significant digits of the vectors.
We analyse all the possible deviations from simulating the counter machine and show that the opponent has a winning strategy in that case. The biggest challenge is to ensure that all possible ways to cheat can be caught without introducing new ways to cheat for the other player.

We prove the main theorem by 
considering the undecidable problem of determining whether a 2CM $\mathcal{M}$ reaches a configuration where both counters are zero. In Section~\ref{sec:2RGS}, we construct a robot game with states that follows the computation of $\mathcal{M}$. 
As the game lacks the ability to perform directly zero checks, instead Adam has a move allowing him to check whether a counter is positive or not leading, deterministically, either to his victory with a correct guess or to his loss otherwise.
In the fourth section, we map the states and state transitions into integers and embed them into the least significant digits in vectors of a two-dimensional robot game. Our proof uses two successive reductions making the proof shorter and more intuitive in contrast to a direct reduction from 2CM that would lead to a longer proof with significantly more cases to consider.


Apart from the solution of the open problem, the main contribution of this paper is a collection of new, original encodings and constructions
that allow simulating zero-checks and state space of a universal machine within a minimalistic two-dimensional system of two non-deterministic stateless players.

{\bf Previous  research:}
Robot games are subfamily of counter reachability games where the game is played on a graph with vertices partitioned between players. It has been proved that deciding the winner in two-dimensional counter reachability games is undecidable \cite{Reichert}. Our result can be seen as strengthening of this as our arena is a graph without self-loops and with one vertex for each player, i.e., both players are stateless.

In \cite{ABO08} and \cite{BJK2010}, VASS games, where the game is played on a graph and counters are always positive, were considered. It was proven that already in two dimensions it is undecidable who wins if the goal is to reach a particular vertex with counter $(0,0)$. On the other hand, if it can be any vertex, then the problem is $(k-1)$-EXPTIME for a game with $k$ counters. Later, the result was improved to PTIME for $k=2$ \cite{C13}. In \cite{Reichert}, the possible counter values were extended to all integers and it was proven that the problem remains undecidable. Hunter considered the variants of games, where updates on the counters are done in binary, and showed that one-dimensional games are EXPSPACE-complete \cite{H15}. While these games have reachability objectives, it is also possible to extend the objectives of the games to energy constrains \cite{FJLS11} or parity constrains \cite{AMSS13,CD12}.

The proofs of undecidability of VASS games and counter reachability in two dimensions in \cite{CSL2003,BJK2010} use the state structure of 
the game to embed the state structure of a 2-counter machine. In this sense, our result on robot games with states is comparable, as Eve simulates the state transitions of a 2-counter machine with her underlying automaton. On the other hand, the stateless game is essentially different as we have to represent state transitions with integers. When simulating a two-counter machine, it is possible for Eve to make a wrong move and then Adam is able to ensure his victory from this point onward. In robot games with states Adam's cheat catching ability is different from when a game is played on a graph. Also in robot games, Eve's state is dependent only on her previous moves, while in VASS games or counter reachability games, Adam's moves effect which state Eve enters.


\section{Notation and Definitions}
We denote the set of all integers by $\Z$ and the set of all non-negative integers by $\N$. By $0_n$ we denote a $n$-dimensional zero vector.

A \emph{counter reachability game} (CRG) consists of a directed graph $G=(V,E)$, where the set of vertices is partitioned into two parts, $V_1$ and $V_2$, each edge $e\in E\subseteq V\times\Z^n\times V$ is labelled with vectors in $\Z^n$, and an initial vector $\bx_0\in\Z^n$. A \emph{configuration} of the game is $(v,\bx)$, a successive configuration is $(v',\bx+\bx')$, where an edge $(v,\bx',v')\in E$ is chosen by player 1 if $v\in V_1$ or by player 2 if $v\in V_2$. A \emph{play} is a sequence of successive configurations. The goal of the first player, called \emph{Eve}, is to reach the \emph{final configuration} $(v_f,0_n)$ for some $v_f\in V$ while the goal of the second player, called \emph{Adam}, is to keep Eve from reaching $(v_f,0_n)$. A \emph{strategy} for a player is a function that maps a configuration to an edge that can be applied. We say that Eve has a \emph{winning strategy} if she can reach the final configuration regardless of the strategies of Adam. On the other hand, we say that Adam has a winning strategy if Eve does not have a winning strategy. In the figures we use $\bigcirc$ for Eve's states and $\Box$ for Adam's states.

A \emph{robot game} (RG) \cite{RobotGames} is a special case of the counter reachability games, where the graph consists of only two vertices, $q_0$ of Adam and $q$ of Eve. The goal of the game is the configuration $(q_0,0_n)$. That is, a robot game consists of two players, \emph{Eve} and \emph{Adam}, having a set of vectors $E$, $A$ over $\Z^n$, respectively, and an \emph{initial vector} $\bx_0$. Starting from $\bx_0$ players add a vector from their respective sets to the current configuration of the game in turns. As in counter reachability games, Eve tries to reach the origin while Adam tries to keep Eve from reaching the origin. The decision problem concerning robot games is, for a given robot game $(A,E)$ and $\bx_0$, to decide whether Eve has a winning strategy to reach $0_n$ from $\bx_0$. The problem is EXPTIME-complete in dimension one \cite{ArulReichert} and was open for dimension two.

An extension of robot games where players have control states is called \emph{robot games with states} (RGS). We consider only the games where Adam's state structure is trivial, i.e., he has only one state and all moves are self-loops. RGS consists of $(A,E)$ where $A$ is a finite subset of $\Z^n$ that Adam can apply during his turn and $E$ is a finite subset of $V\times \Z^n \times V$ of Eve. The configuration is now a pair $(s,\bv)$ consisting of Eve's control state $s$ and a counter vector $\bv\in\Z^n$. Eve updates her control state when she makes a move: in the configuration $(s,\bv)$, for any vector $\bv$, only moves of the form $(s,\bx,t)$ are enabled, and with one such move the new configuration is $(t,\bv + \bx)$. Eve wins if, and only if, after her turn, the configuration is $(s,(0,0))$ for any $s\in V$. The decision problem associated with robot games with states asks whether Eve has a winning strategy from a given configuration. 

A \emph{Minsky machine}, introduced in \cite{M67}, is a simple computation model that is crucial in our proof.  A \emph{deterministic two-counter Minsky machine} (2CM) is a pair $(Q,T)$, where $Q$ is a finite set of states and 
$T\subseteq Q\times\{c_i{\scriptscriptstyle ++},c_i{\scriptscriptstyle --},c_i{\scriptstyle ==}0\mid i=1,2\} \times Q$ is a finite set of labelled transitions to increment, decrement or test for zero one of the counters. In a deterministic two-counter Minsky machine, the set $Q$ contains an initial state $s_0$ and a sink state $\bot$, such that there is no outgoing transition from $\bot$. Moreover, from all $s\in Q\setminus\{\bot\}$, either there is only one outgoing transition with the label $c_1{\scriptscriptstyle ++}$ or $c_2{\scriptscriptstyle ++}$, or there are exactly two outgoing transitions with respective labels $c_1{\scriptscriptstyle --}$ and $c_1{\scriptstyle ==}0$, or $c_2{\scriptscriptstyle --}$ and $c_2 {\scriptstyle ==}0$. A configuration of a 2CM is a pair $(s,(y,z))\in Q\times\N^2$, representing a state and a pair of counter values. The run of a 2CM is a finite or infinite sequence of configurations that starts from $(s_0 ,(0,0))$ and follows the transitions of the machine incrementing and decrementing the counters according to the labels. As usual, a transition with a label $c_i{\scriptstyle ==}0$ can only be taken when the counter $i$ is zero and a transition with a label $c_i{\scriptscriptstyle --}$ can only be taken when the counter $i$ is positive.

Note that there is only one possible run in a deterministic two-counter Minsky machine. Indeed, when there are two outgoing transitions, only one of them can be executed, depending on the value of the counter that the transitions update or test for zero. The halting problem of 2CM is to decide, given a 2CM, whether the run reaches a configuration with state $\bot$, in other words whether the run halts. This problem is known to be undecidable for deterministic two-counter machines \cite{M67}.
Another well-known undecidable problem for 2CM is whether a configuration where both counters are zero is reachable. The undecidability follows from the halting problem by modifying a 2CM to ensure that both counters are zero only in the halting state; see for example \cite{Reichert-phd15} for a proof.

\setcounter{reach}{\value{theorem}}
\begin{theorem}\label{Reach2CM}
Let $(Q,T)$ be a deterministic two-counter machine. It is undecidable whether in the run of $(Q,T)$, a configuration in $Q\times\{(0,0)\}$ appears.
\end{theorem}

We can assume that the first move of a 2CM is an increment of either $c_1$ or $c_2$. Indeed, otherwise the problem is trivial as the second configuration is in $Q\times\{(0,0)\}$.

\section{Robot games with states in two dimensions}\label{sec:2RGS}
In this section we prove that the decision problem for robot games with states is undecidable. We show that for each two-counter machine, there exists a corresponding robot game with states where Eve has a winning strategy if and only if the machine reaches a configuration where both counters are zero. The game lacks the ability to perform zero checks present in two-counter machines, instead Adam has a move allowing him to check whether a counter is positive or not.

\begin{theorem}\label{2CMiff2RGS}
Let $(Q,T)$ be a two-counter machine. There exists a two-dimensional robot game with states $(A,E)$ where Eve has a winning strategy if and only if $(Q,T)$ reaches a configuration in $Q\times\{(0,0)\}$.
\end{theorem}

The idea is that in the robot game with states, Eve simulates the computation of the 2CM while Adam does not interfere with the computation. If one of the players deviates from the computation, the opponent has a winning strategy from that point on. 

Essentially, there are four ways the game can progress. These ways are depicted in the Figure~\ref{2rgs4ways}. Three of the outcomes have a predetermined winner which does not depend on the 2CM. In the last case where Eve correctly simulates the 2CM and Adam does not interfere (plays only a \textsc{0-move}), the winner depends on whether the 2CM reaches $(q,(0,0))$ for some $q\in Q$ or not.
\begin{itemize}
\item If Eve's move corresponds to the \textsc{simulation} of the 2CM and Adam replies with a \textsc{0-move} (a move that does not modify the counters), then iteratively applying only this turn-based interaction, Eve wins if and only if the 2CM reaches $(q,(0,0))$ for some $q\in Q$ (Lemma~\ref{2rgshalt}).
\item If Eve's move incorrectly simulates the 2CM, then Adam has a winning strategy from this moment on, starting with a \textsc{positivity check} that makes Eve's target unreachable (Lemma~\ref{2rgsevecheats}).
\item On the other hand, if Adam plays his \textsc{positivity check} following a correct simulating move of Eve, then Eve has a winning strategy from this moment on, starting with an \textsc{emptying move} allowing Eve to empty both counters and reach $(0,0)$ (Lemma~\ref{2rgsadamcheats}).
\item This leads to the possibility that Eve plays an \textsc{emptying move} instead of a \textsc{simulating move}, in that case Adam has a winning strategy starting by playing his \textsc{0-move} (Lemma~\ref{2rgseveresponds}).
\end{itemize}
\begin{figure}[htb]
\centering
\begin{tikzpicture}[scale=.8, every node/.style={scale=0.8}]
\path [fill=Grayish] (-0.5,-0.94) rectangle (9.1,1.5);
\draw [dotted] (-0.5,-0.94) -- (9.1,-0.94) -- (9.1,1.5) -- (-0.5,1.5) -- (-0.5,-0.94);
\node [circle,draw] (Eve) {$E$};
\node [rectangle,draw] (Adam)[right = 2cm of Eve] {$A$};
\node [rectangle,draw] (wrong)[left = 2cm of Eve] {$A$};
\node (AdamW1) [below = 2.7cm of wrong] {Adam wins};
\node [rectangle,draw](empty) [below = 1.2cm of Eve] {$A$};
\node (AdamW2) [below = 2.6cm of Eve] {Adam wins};
\node [circle,draw] (emptying) [below =1.3cm of Adam] {$E$};
\node (EveW) [below = 2.7cm of Adam] {Eve wins};

\draw (Eve) edge  [-stealth, bend left=30] node[align=center,above]{simulation\\(correct)} (Adam);
\draw (Adam) edge  [-stealth, dashed, bend left=30] node[below]{0-move} (Eve);
\draw (Eve) edge  [-stealth, bend right=30] node[align=center,above]{simulation\\(incorrect)} (wrong);
\draw (wrong) edge  [-stealth, dashed] node[align=right, left]{positivity\\check} (AdamW1);
\draw (Eve) edge  [-stealth] node[align=right, left, pos=0.65]{emptying\\move} (empty);
\draw (empty) edge  [-stealth, dashed] node[align=right, left]{0-move} (AdamW2);
\draw (Adam) edge  [-stealth,dashed] node[align=left, right, pos=0.67]{positivity\\check} (emptying);
\draw (emptying) edge  [-stealth,] node[align=left, right,pos=0.4]{emptying\\move} (EveW);
\node [align=right,right] at (6,-2.6) {\parbox{3.7cm}{Eve's moves:
\begin{itemize}
\item \textsc{simulation} of 2CM (correct/incorrect)
\item \textsc{emptying move}
\end{itemize}
Adam's moves:
\begin{itemize}
\item \textsc{0-move}
\item \textsc{positivity check}
\end{itemize}}};
\node [align=left] at (6.3,0.3) {Eve wins if \\ \qquad 2CM reaches $Q\times(0,0)$ \\ Adam wins if \\ \qquad2CM does not reach $Q\times(0,0)$};
\end{tikzpicture}
\caption{\label{2rgs4ways}Progress of 2RGS}
\end{figure}

Before presenting the detailed constructions of Eve's and Adam's state spaces, we consider a simple modification to a 2CM, making it non-deterministic. For any 2CM $(Q,T)$, we construct a 2CM $(Q',T')$ where $Q'$ is $Q$ with additional information on positivity of the both counters and $T'$ is like $T$ with guards ensuring that the extra information in states of $Q'$ correspond to the actual values of the counters. We denote the states of $Q'$ by $s_{ab}$ where $a,b\in\{0,+\}$ are flags indicating whether the value of a counter is positive or equal to 0, i.e., $a$ ($b$) is $+$ if the first (second) counter is positive or $0$ if the counter is zero. The transition set $T'$ consists of the following sets
{\small\begin{align*}
&\hspace*{-6.5pt}\left\{(s_{ab},c_1{\scriptscriptstyle ++},t_{+b})\mid (s,c_1{\scriptscriptstyle ++},t)\in T, a,b\in\{0,+\}\right\}, \left\{(s_{ab},c_2{\scriptscriptstyle ++},t_{a+})\mid (s,c_2{\scriptscriptstyle ++},t)\in T, a,b\in\{0,+\}\right\}, \\
&\hspace*{-7.5pt}\left\{(s_{+b},c_1{\scriptscriptstyle --},t_{ab})\mid (s,c_1{\scriptscriptstyle --},t)\in T, a,b\in\{0,+\}\right\} ,\left\{(s_{a+},c_2{\scriptscriptstyle --},t_{ab})\mid (s,c_2{\scriptscriptstyle --},t)\in T, a,b\in\{0,+\}\right\}, \\
&\hspace*{-6.5pt}\left\{(s_{0b},c_1{\scriptstyle ==}0,t_{0b})\mid (s,c_1{\scriptstyle ==}0,t)\in T, b\in\{0,+\}\right\}, \left\{(s_{a0},c_2{\scriptstyle ==}0,t_{a0})\mid (s,c_2{\scriptstyle ==}0,t)\in T, a\in\{0,+\}\right\}.
\end{align*}}
Now, after decrementing counters from a state with $+$ flag, a state will changed to a state with $+$ or 0 flag depending on the current counter value. 
\begin{center}
\begin{tabular}{ccccl}
counter value & flag & & flag & \\
\hline
\rowcolor{Gray}$c_i>1$ & $+$ & $\to$ & $+$ & \textbf{correct flag} \\
$c_i>1$ & $+$ & $\to$ & 0 & wrong flag\\
$c_i=1$ & $+$ & $\to$ & + & wrong flag\\
\rowcolor{Gray}$c_i=1$ & + & $\to$ & 0 & \textbf{correct flag}\\
\end{tabular}
\end{center}
At the moment we assume that the machine moves to a state with the correct flag (correct simulation) and does not move to incorrect flag (incorrect simulation). Later in the robot game with states, Adam will act as guards (i.e., checks whether $c_i>1$ or $c_i=1$) using his \textsc{positivity check} if Eve picks a wrong transition resulting in a state with the wrong flag.

Now we present the moves of the players. Eve's states are the states of $Q'$, corresponding to the simulation of the 2CM, together with emptying states $\{\top_{00},\top_{+0},\top_{0+},\top_{++}\}$, associated with \textsc{emptying moves}. The moves of Eve correspond to transitions in $T'$ where incrementing and decrementing of the first counter is by 4 rather than by 1.  We call these moves \textsc{simulating moves}, see Figure~\ref{statettrans} in the Appendix:
\begin{center}
\begin{tabular}{|c|c|}
\hline
Transition with $c_1$ & Eve's move \\
\hline
$(s,c_1 {\scriptstyle ++},t)$ & $(s,(4,0),t)$ \\
\hline
$(s,c_1 {\scriptstyle --},t)$ & $(s,(-4,0),t)$ \\
\hline
$(s,c_1 {\scriptstyle ==}0,t)$ & $(s,(0,0),t)$ \\
\hline
\end{tabular} \qquad
\begin{tabular}{|c|c|}
\hline
Transition with $c_2$ & Eve's move \\
\hline
$(s,c_2 {\scriptstyle ++},t)$ & $(s,(0,1),t)$ \\
\hline
$(s,c_2 {\scriptstyle --},t)$ & $(s,(0,-1),t)$ \\
\hline
$(s,c_2 {\scriptstyle ==}0,t)$ & $(s,(0,0),t)$ \\
\hline
\end{tabular}
\end{center}
The other type of moves, \textsc{emptying moves}, are related to the new states and are used to empty the counters. 
Note that there is hierarchy in the emptying states --- Eve cannot move from a state with $0$ to a state with $+$. Let us define the emptying partition of Eve's automaton where for every possible move of Adam there is a cancelling move with additional decrementing of the counters eventually leading to the sink state $\top_{00}$.
\begin{itemize}
\item $\{(\top_{++} ,(-4 - e,-1),t) \mid e \in\{0,1\},t\in\{\top_{++},\top_{+0},\top_{0+},\top_{00}\}\}$;
\item $\{(\top_{+0} ,(-4 - e,0),t) \mid e \in\{0,1\},t\in\{\top_{+0},\top_{00}\}\}$;
\item $\{(\top_{0+} ,(-e,-1),t ) \mid e \in\{0,1\},t\in\{\top_{0+},\top_{00}\}\}$;
\item $\{(\top_{00} ,(- e,0),\top_{00} ) \mid e \in\{0,1\}\}$.
\end{itemize}
Finally, we define transitions connecting the simulating partition of Eve's automaton with the emptying partition. For each state $s_{ab}\in Q'$, Eve has a transition $(s_{ab},(-1,0),\top_{ab})$.

Adam is stateless, i.e., he has one state and his moves are self-loops. There are two types of moves: the \textsc{0-move}, $(0,0)$, with which Adam agrees that Eve simulated the 2CM correctly and the \textsc{positivity check}, $(1,0)$, with which Adam checks whether a flag matches the counter (i.e., Eve simulated incorrectly). 
Control states of the players are depicted in Figure~\ref{eveadamstates}.

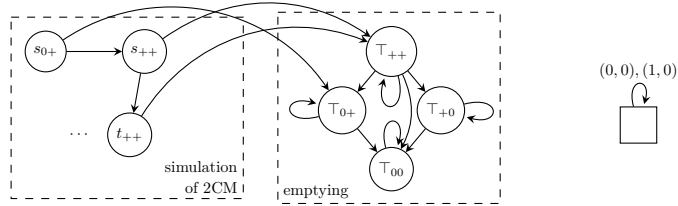
\begin{figure}[htb]
\centering
\begin{tikzpicture}[scale=.65,>=stealth, every node/.style={scale=0.6}]
  \node[state] 		   (s0+) at (0,0)  {$s_{0+}$};
  \node[state]         (s++) at (2,0) {$s_{++}$};
  \node[state]         (t++) at (1.7,-1.7)       {$t_{++}$};
  \node[state]         (top++) at (7,0)       {$\top_{++}$};
  \node[state]         (top0+) at (6,-1.2)       {$\top_{0+}$};
  \node[state]         (top+0) at (8,-1.2)       {$\top_{+0}$};
  \node[state]         (top00) at (7,-2.4)       {$\top_{00}$};
  \path[->] (s0+) edge              (s++)
                  (s++) edge              (t++)
                  (top++) edge            (top+0)
                  		edge [loop below] ()
						edge            (top0+)
                  (top+0) edge            (top00)
                  		edge [loop right] ()
                  (top0+) edge            (top00)
                  		edge [loop left] ()
                  (top00) edge [loop above] ()
                  (s0+) edge [bend left=45] (top0+)
                  (s++) edge [bend left=35] (top++)
                  (t++) edge [bend left=40] (top++)
                  (top++) edge [bend left=25.5] (top00);
\draw [dashed] (-0.7,0.7) -- (4,0.7) -- (4,-3) -- (-0.7,-3) -- (-0.7,0.7);
\draw [dashed] (4.7,0.8) -- (9.2,0.8) -- (9.2,-3.1) -- (4.7,-3.1) -- (4.7,0.8);
\draw node[align=right,above left] at (4,-3) {simulation\\of 2CM};
\draw node[align=left,above right] at (4.7,-3.1) {emptying};
\draw node at (0.7,-1.7) {$\cdots$};

\node [rectangle,draw, minimum size=0.8cm] (Adam) at (12,-1.5) {};
  \path[-stealth] (Adam) edge [loop above] node {$(0,0),(1,0)$}            (s++);
\end{tikzpicture}
\caption{\label{eveadamstates}An illustration of state transitions of Eve and Adam}
\end{figure}

To avoid Eve winning trivially every play in the robot game with states, we do not use $(s'_{00},(0,0))$ as an initial configuration, but instead consider the configuration that is reached in $(Q',T')$ after one step of the run of the machine. We write the configuration after one step $(s_{\overline{a}\overline{b}},(y,z))$ and we define $\overline{a} = +,\overline{b}=0$ if $y=1$ and $\overline{a}=0,\overline{b}= +$ if $y=0$. The initial configuration in the robot game with states is then $(s_{\overline{a}\overline{b}},(4y,z))$.
The effect of simulating moves, emptying moves and positivity check modulo four is depicted in Figure~\ref{2RGSinterval}.

\begin{figure}[htb]
\centering
\begin{tikzpicture}[xscale=1.5,yscale=0.75]
\draw (0,0)--(8,0);
\draw (0,0.15) -- (0,-0.15);
\draw (1,0.1) -- (1,-0.1);
\draw (2,0.1) -- (2,-0.1);
\draw (3,0.1) -- (3,-0.1);
\draw (4,0.15) -- (4,-0.15);
\draw (5,0.1) -- (5,-0.1);
\draw (6,0.1) -- (6,-0.1);
\draw (7,0.1) -- (7,-0.1);
\draw (8,0.15) -- (8,-0.15);
\draw[decorate,decoration={brace,amplitude=12pt}] (0,0) -- (4,0) node[above, pos=0.5,yshift=0.4cm]{$4$};
\draw[decorate,decoration={brace,amplitude=12pt}] (4,0) -- (8,0) node[above, pos=0.5,yshift=0.4cm]{$4$};
\draw[shorten >=2pt,-stealth] (0,0) to [bend right =50] node[below, pos=0.5]{simulating move} (4,0) ;
\draw[shorten >=2pt,-stealth,dashed] (4,0) to [bend left =40] node[below left, pos=0.5,yshift=0.13cm]{positivity check} (3,0) ;
\draw[shorten >=2pt,-stealth] (4,0) to [bend right =50] node[below, pos=0.5]{emptying move} (8,0) ;
\end{tikzpicture}
\caption{An illustration of changes in an interval when simulating or emptying moves of Eve or positivity check of Adam is applied \label{2RGSinterval}}
\end{figure}
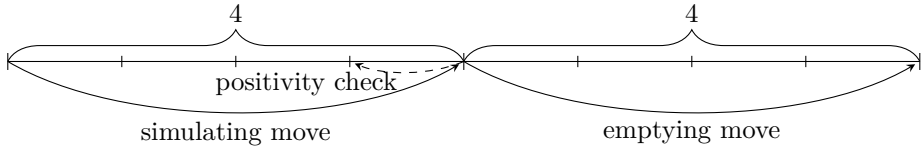

Next we prove which player has a winning strategy in the scenarios presented previously.

\begin{lemma}\label{2rgshalt}
In a sequence where Adam plays only the \textsc{0-move} and Eve plays only correct \textsc{simulating moves}, 
Adam wins if the 2-counter machine does not reach a configuration with zeros in both counters and Eve wins otherwise.
\end{lemma}
\begin{proof}
It easy to see that correct moves of Eve simulate the 2CM and that a configuration $(s,(0,0))$ of the 2CM is reachable if and only if it is reachable in 2RGS. 
\end{proof}

\begin{lemma}\label{2rgsevecheats}
If Eve plays an incorrect move, i.e., after her turn a flag does not match the counter value (i.e., the flag is $+$ while the counter is 0 or vice versa), Adam has a winning strategy starting with the \textsc{positivity check}.
\end{lemma}
\begin{proof}
Assume that Eve made a mistake regarding the positivity of the first counter. As noted previously, there are two ways she can make a mistake. Either the configuration is $(s_{0b},(4x,y))$, where $x\geq1$ or $(s_{+b},(0,y))$. In both cases Adam plays his \textsc{positivity check} which changes the parity of the first counter. That is, after Adam's turn, the first counter is $1\pmod4$. It is easy to see that if Eve does not change the parity of the counter back to zero with her following turn, then Adam has a winning strategy. Indeed, he will play his \textsc{positivity check} if and only if the first counter is not $3\pmod4$. Eve cannot make the counter 0, as she cannot even make it $0\pmod4$. Thus Eve has to play a move adding $-1$ to the first counter. The only move for that is $(s_{ab},(-1,0),\top_{ab})$ which takes Eve to an emptying state. In the first case the emptying state is $\top_{0b}$ and all the transitions from it do not modify the first counter, i.e., Eve cannot reach $(0,0)$. In the second case the emptying state is $\top_{+b}$ where the next transition subtracts 4 from the first counter making it negative and there are no moves that increment the counters. Again, Eve cannot reach $(0,0)$. The case where Eve makes a mistake with the second counter is symmetric and is proven analogously.
\end{proof}

\begin{lemma}\label{2rgsadamcheats}
If Eve plays only correct \textsc{simulating moves} before Adam plays the \textsc{positivity check} for the first time, then Eve has a winning strategy starting with an \textsc{emptying move}.
\end{lemma}
\begin{proof}
Similarly as in the previous proof, if Eve does not play an \textsc{emptying move}, then Adam has a winning strategy. Now, the configuration is $(s_{ab},(4x+1,y))$ after Adam's turn and Eve plays $(s_{ab},(-1,0),\top_{ab})$. From that point onward, Eve can empty the counters ensuring that the first counter is $0\pmod4$ and that the flags match the positivity of the counters. That is, every time Adam plays his \textsc{positivity check}, Eve plays an \textsc{emptying move} subtracting one from the first counter. Eventually, Eve will reach the configuration $(\top_{00},(0,0))$ and win the game.
\end{proof}

\begin{lemma}\label{2rgseveresponds}
If Adam plays only the \textsc{0-move} and Eve plays an \textsc{emptying move}. Adam has a winning strategy starting with the \textsc{0-move}.
\end{lemma}
\begin{proof}
After Eve's move, the first counter is $3\pmod4$. As in proof of Lemma~\ref{2rgsevecheats}, Adam ensures that the first counter stays non-zero modulo four and wins the game.
\end{proof}

\begin{proof}[Proof of Theorem~\ref{2CMiff2RGS}]
Let $(A,E)$ be the robot game with states constructed in this section. Assume first that $(Q,T)$ reaches a configuration in $Q\times\{(0,0)\}$. Now by Lemma~\ref{2rgshalt}, Eve's winning strategy is to respond with the correct \textsc{simulating moves} if Adam plays the \textsc{0-move}, and if Adam plays a \textsc{positivity check}, then Eve has a sequence of moves described in Lemma~\ref{2rgsadamcheats} that leads to the configuration $(\top_{00},(0,0))$.

Assume then that $(Q,T)$ never reaches a configuration in $Q\times\{(0,0)\}$. We show that Eve does not have a winning strategy. If Adam plays only the \textsc{0-move}, then, by Lemma~\ref{2rgshalt}, Eve does not win by responding with just the correct \textsc{simulating moves}. Alternatively, if at some point, she plays either an incorrect \textsc{simulating move} or an \textsc{emptying move}, then by Lemmas~\ref{2rgsevecheats} and \ref{2rgseveresponds}, respectively, Adam has winning strategies making sure that a configuration with counter values $(0,0)$ is not reachable. As we analysed all the possible moves of Eve, we have shown that Eve does not have a winning strategy.
\end{proof}
By Theorems~\ref{Reach2CM} and \ref{2CMiff2RGS}, we have the following corollary regarding decidability of 2-dimensional robot games with states.

\begin{corollary}\label{2RGS}
Let $(A,E)$ be a robot game with states and $\bx_0$ be the initial vector. It is undecidable whether Eve has a winning strategy to reach $(0,0)$ from $\bx_0$. In particular, Adam is stateless and does not modify the second counter.
\end{corollary}

\section{Stateless robot games in two dimensions}\label{sec:2RG}
In this section we prove the main result that it is undecidable whether Eve has a winning strategy in a two-dimensional robot game. We prove the claim by constructing a robot game that simulates a robot game with states. In some ways the construction is similar to the construction of a game with states in the previous section as can be seen in similarities of figures \ref{2rgs4ways} and \ref{2rg4ways}. On the other hand, the construction of the stateless game is more complex as the information on two counters, states and state transitions has to be embedded into two-dimensional vectors.

\begin{theorem}\label{2RGSiff2RG}
Let $(A_1,E_1)$ be a 2-dimensional robot game with states where Adam is stateless and does not modify the second counter. There exists a two-dimensional robot game $(A,E)$ where Eve has a winning strategy if and only if Eve has a winning strategy in $(A_1,E_1)$.
\end{theorem}

Similarly to the construction of Section~\ref{sec:2RGS}, the idea is that in the robot game, Eve and Adam simulate a play of the 2RGS. If one of the players deviates from the play, the opponent has a winning strategy from that point onward. 
In Figure~\ref{2rg4ways}, we present a schematic similar to Figure~\ref{2rgs4ways} depicting the possible ways two-dimensional robot games can go. Three of the outcomes have a predetermined winner which does not depend on the 2RGS. In the last case where Eve and Adam correctly simulate the 2RGS, the winner depends on the winner of the 2RGS, i.e., whether Eve has a winning strategy to reach $(s,(0,0))$, for any state $s$, or not.
\begin{itemize}
\item If Eve's move corresponds to a move in a play of the 2RGS, that we call a \textsc{regular move}, and Adam replies with his \textsc{regular move}, then iteratively applying only this turn-based interaction, Eve has a winning strategy if and only if she has a winning strategy in the corresponding 2RGS (Lemma~\ref{2rgwin}).
\item If Eve's move incorrectly simulates the 2RGS, then Adam has a winning strategy from this moment on starting with a \textsc{state-check} that makes Eve's target unreachable (Lemma~\ref{2rgevecheats}). 
\item On the other hand, if Adam plays his \textsc{state-check} following a correct \textsc{regular move} of Eve, then Eve has a winning strategy from this moment on starting with a \textsc{state-defence move} allowing Eve to empty both counters and reach $(0,0)$ (Lemma~\ref{2rgadamcheats}).
\item This leads to the possibility that Eve plays a \textsc{state-defence move} instead of a \textsc{regular move}, in that case Adam has a winning strategy starting by playing his \textsc{regular move} (Lemma~\ref{2rgeveresponds}).
\end{itemize}

\begin{figure}[htb]
\centering
\begin{tikzpicture}[scale=.8, every node/.style={scale=0.8}]
\path [fill=Grayish] (-0.5,-1.0) rectangle (8.1,1.5);
\draw [dotted] (-0.5,-1.0) -- (8.1,-1.0) -- (8.1,1.5) -- (-0.5,1.5) -- (-0.5,-1.0);
\node [circle,draw] (Eve) {$E$};
\node [rectangle,draw] (Adam)[right = 2cm of Eve] {$A$};
\node [rectangle,draw] (wrong)[left = 2cm of Eve] {$A$};
\node (AdamW1) [below = 2.7cm of wrong] {Adam wins};
\node [rectangle,draw](empty) [below = 1.2cm of Eve] {$A$};
\node (AdamW2) [below = 2.6cm of Eve] {Adam wins};
\node [circle,draw] (emptying) [below =1.2cm of Adam] {$E$};
\node (EveW) [below = 2.7cm of Adam] {Eve wins};

\draw (Eve) edge  [-stealth, bend left=30] node[align=center,above]{simulation\\(correct)} (Adam);
\draw (Adam) edge  [-stealth, dashed, bend left=30] node[below]{regular move} (Eve);
\draw (Eve) edge  [-stealth, bend right=30] node[align=center,above]{simulation\\(incorrect)} (wrong);
\draw (wrong) edge  [-stealth, dashed] node[align=right, left]{state-\\check} (AdamW1);
\draw (Eve) edge  [-stealth] node[align=right, left,pos=0.67]{state-defence\\move} (empty);
\draw (empty) edge  [-stealth, dashed] node[align=right, left]{regular\\move} (AdamW2);
\draw (Adam) edge  [-stealth,dashed] node[align=left, right,pos=0.65]{state-check} (emptying);
\draw (emptying) edge  [-stealth,] node[align=left, right,pos=0.4]{state-defence\\move} (EveW);
\node [align=right,right] at (6,-2.65) {\parbox{4.05cm}{Eve's moves:
\begin{itemize}
\item \textsc{simulation} of 2RGS (correct/incorrect)
\item \textsc{state-defence move}
\end{itemize}
Adam's moves:
\begin{itemize}
\item \textsc{regular move}
\item \textsc{state-check}
\end{itemize}}};
\node [align=left] at (6,0.3) {Eve wins if \\ \qquad Eve wins in 2RGS \\ Adam wins if \\ \qquad Adam wins in 2RGS};
\end{tikzpicture}

\caption{\label{2rg4ways}Progress of 2RG}
\end{figure}
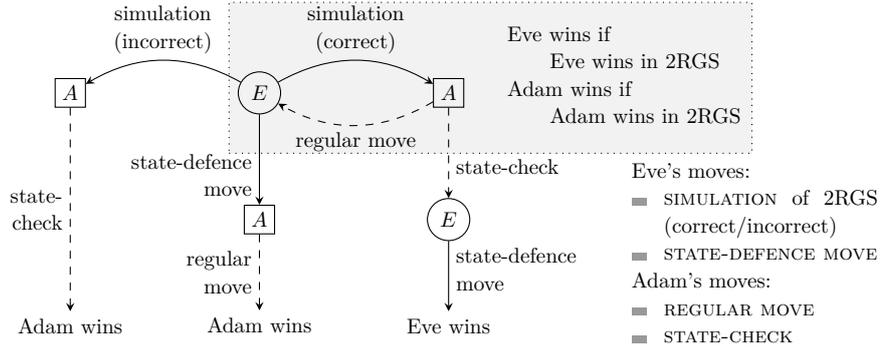

Intuitively, we encode the states as powers of $8$ such that the coefficient of $8^i$ is $1$ if and only if Eve's state in robot games with states  is $s_i$. When the state changes from $s_i$ to $s_j$, $-8^i+8^j$ is added to the second counter. Let $(s_i,(x,y))$ be a configuration in a two-dimensional robot game with $m$ states. Let us represent the state $s_i$ with $m$-dimensional characteristic vector $\mathbf{s_i}=(s_1,\ldots,s_m)$ where $s_i$ is 1 and $s_j=0$ for all $j\neq i$. We can now map $\mathbf{s_i}$ to an integer defined by the sum $\sum_{k=1}^m s_k8^k$. A transition from $s_i$ to $s_j$ can be simulated by adding $\sum_{k=1}^m s_k8^k$, where $s_k=-1$ if $k=i$, $s_k=1$ if $k=j$, and zero otherwise. Note that we represent states as coefficients of powers of eight because we need the extra space smaller bases do not possess. 

It is easy to see that this is not enough as incorrect transitions can result in a correct configuration. For example, if the configuration of the 2RGS is $(s_i,(x,y))$ and moves corresponding to $(s_j,(a,b),s_k)$ and $(s_k,(c,d),s_j)$ are used, the resulting configuration corresponds to $(s_i,(x+a+c,y+b,d))$. 
Another way to cheat is to use carries
as incrementing the coefficient of $8^i$ eight times is indistinguishable from incrementing the coefficient of $8^{i+1}$ once.
Both types of cheating can be countered with Adam's \textsc{state-checks}. 

We now show how we embed the states and state transitions into the second counter of the game. Similarly to how in the previous section we created additional space in the first counter by multiplying the moves modifying the first counter by four, we multiply the second counter by $4\cdot8^n$, where $n=m+7$ and $m$ is the number of states, creating enough space to store all the needed information of the underlying automaton. The multiplication by $4\cdot8^n$ rather than just $8^n$ has two purposes. The first one is similar to multiplying the first counter by four in the Section~\ref{sec:2RGS}. Namely, certain moves will move between different intervals modulo $4\cdot8^n$ ensuring the correct response from the opponent. 
This is illustrated in Figure~\ref{2RGinterval}. 
The second purpose is to ensure that above described cheating with carries is not possible.
A configuration in $Q\times\Z^2$ is mapped to a vector in $\Z^2$ by 
$
(s_i,(c_1,c_2))\mapsto(c_1,c_2\cdot4\cdot8^n+8^{i}).
$

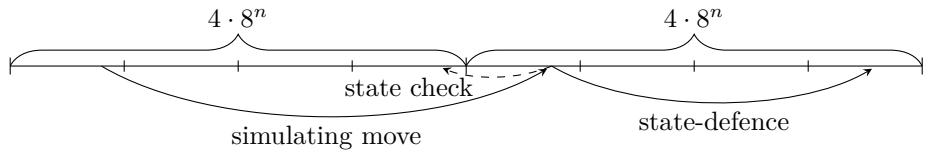
\begin{figure}[htb]
\centering
\begin{tikzpicture}[xscale=1.5,yscale=0.75]
\draw (0,0)--(8,0);
\draw (0,0.15) -- (0,-0.15);
\draw (1,0.1) -- (1,-0.1);
\draw (2,0.1) -- (2,-0.1);
\draw (3,0.1) -- (3,-0.1);
\draw (4,0.15) -- (4,-0.15);
\draw (5,0.1) -- (5,-0.1);
\draw (6,0.1) -- (6,-0.1);
\draw (7,0.1) -- (7,-0.1);
\draw (8,0.15) -- (8,-0.15);
\draw[decorate,decoration={brace,amplitude=12pt}] (0,0) -- (4,0) node[above, pos=0.5,yshift=0.4cm]{$4\cdot 8^n$};
\draw[decorate,decoration={brace,amplitude=12pt}] (4,0) -- (8,0) node[above, pos=0.5,yshift=0.4cm]{$4\cdot 8^n$};
\draw[shorten >=2pt,-stealth] (0.8,0) to [bend right =50] node[below, pos=0.5]{simulating move} (4.75,0) ;
\draw[shorten >=2pt,-stealth,dashed] (4.75,0) to [bend left =40] node[below left, pos=0.6,yshift=0.1cm]{state check} (3.75,0) ;
\draw[shorten >=2pt,-stealth] (4.75,0) to [bend right =50] node[below, pos=0.5]{state-defence} (7.6,0) ;
\end{tikzpicture}
\caption{An illustration of changes in interval when simulating or state-defence moves of Eve or state check of Adam is applied \label{2RGinterval}}
\end{figure}

Before presenting the detailed constructions of Eve's and Adam's moves, we note that we can assume that the 2RGS has the information on the positivity of the counters and players have to update the information correctly. Indeed, this was done in the previous section by using flags $0$ and $+$. Recall that because of this, the first counter is incremented and decremented by 4. By this assumption, we can denote the states of Eve by $s_{ab}$ as before. We also assume that Eve's automaton is without self-loops. Let $Q$ be the set of states of Eve in 2RGS. We create an emptying gadget for Eve similar to the one constructed in the previous reduction. To avoid self-loops, there are seven emptying states, $\{\top_{ab},\top'_{ab}\mid a,b\in\{0,+\}\}\setminus\{\top'_{00}\}$. The state $\top'_{00}$ is not needed as $\top_{00}$ will not have any moves from it. The moves in the emptying gadget are
as in the emptying gadget constructed in Section~\ref{sec:2RGS} but instead of a self-loops, the transitions are between primed and unprimed versions of the states.
\begin{itemize}
\item $\{(\top_{++} ,(-4,-1)-\alpha,t) \mid \alpha \in A_1,t\in\{\top'_{++},\top_{+0},\top'_{+0},\top_{0+},\top'_{0+},\top_{00}\}\}$; \\
 $\{(\top'_{++} ,(-4,-1)-\alpha,t) \mid \alpha\in A_1,t\in\{\top'_{++},\top_{+0},\top'_{+0},\top_{0+},\top'_{0+},\top_{00}\}\}$;
\item $\{(\top_{+0} ,(-4,0)-\alpha,t) \mid \alpha \in A_1,t\in\{\top'_{+0},\top_{00}\}\}$; \\
$\{(\top'_{+0} ,(-4,0)-\alpha,t) \mid \alpha \in A_1,t\in\{\top_{+0},\top_{00}\}\}$;
\item $\{(\top_{0+} ,(0,-1)-\alpha,t) \mid \alpha\in A_1,t\in\{\top'_{+0},\top_{00}\}\}$; \\
$\{(\top'_{0+} ,(0,-1)-\alpha,t) \mid \alpha\in A_1,t\in\{\top_{+0},\top_{00}\}\}$;
\end{itemize}

We denote $\mathcal{T}=\{\top_{++},\top'_{++},\top_{0+},\top'_{0+},\top_{+0},\top'_{+0}\}$.
We think of elements of $Q\cup\mathcal{T}\cup\{\top_{00}\}$ as integers in $\{0,\ldots,n-1\}$ such that $\top_{00}=0,\top_{0+}'=n-6,\top_{0+}=n-5,\top_{+0}'=n-4,\top_{+0}=n-3,\top_{++}'=n-2,\top_{++}=n-1$. We give names for update vectors that we often use:
\begin{alignat*}{2}
\add(1,x)&:=(x,0); & \move(j,k)&:=(0,-8^{j}+8^{k})\text{, for }0\leq j,k\leq n-1;\\
\add(2,x)&:=(0,4x\cdot8^{n}); &
\ch(i)&:=(0,-5\cdot8^{i}-8^{n})\text{, for }n-6\leq i\leq n-1.
\end{alignat*}
 The initial vector of the robot game is 
$\add(1,x)+\add(2,y)+\move(\top_{00},s)$, that is, $(x,4y\cdot8^{n}+8^{s}-8^0)$,
where $(s,(x,y))$ is the initial configuration in the robot game with states.
In the next example we illustrate how the update vectors modify the counters.
\begin{example}
Let $(A_1,E_1)$ be a two-dimensional robot game with states where Eve has two states, $s=1$ and $t=2$, and the initial configuration $(s,(1,0))$. Next we present a set of configurations in 2RG obtained from the corresponding initial configuration when we apply $\add(1,-1)$, $\add(2,1)$, $\move(s,t)$, $\ch(8)$ in succession:
{\small \begin{align*}
&\overbrace{(1,0\cdot4\cdot8^9}^{\text{2RGS counters}} + \overbrace{0\cdot8^8 + 0\cdot8^7 + 0\cdot8^6 + 0\cdot8^5 + 0\cdot8^4 + 0\cdot8^3}^{\T} + \overbrace{0\cdot8^2 + 1\cdot8^1}^{\text{states of 2RGS}} \overbrace{-1\cdot8^0)}^{\top_{00}}
\xlongrightarrow{\add(1,-1)}  \\
&(0,0\cdot4\cdot8^9 + 0\cdot8^8 + 0\cdot8^7 + 0\cdot8^6 + 0\cdot8^5 + 0\cdot8^4 + 0\cdot8^3 + 0\cdot8^2 + 1\cdot8^1 - 1\cdot8^0)
\xlongrightarrow{\add(2,1)} \\
&(0,4\cdot8^9 + 0\cdot8^8 + 0\cdot8^7 + 0\cdot8^6 + 0\cdot8^5 + 0\cdot8^4 + 0\cdot8^3 + 0\cdot8^2 + 1\cdot8^1 - 1\cdot8^0) 
\xlongrightarrow{\move(s,t)}  \\
&(0,4\cdot8^9 + 0\cdot8^8 + 0\cdot8^7 + 0\cdot8^6 + 0\cdot8^5 + 0\cdot8^4 + 0\cdot8^3 + 1\cdot8^2 + 0\cdot8^1 - 1\cdot8^0) 
\xlongrightarrow{\ch(8)}  \\
&(0,3\cdot8^9 - 5\cdot8^8 + 0\cdot8^7 + 0\cdot8^6 + 0\cdot8^5 + 0\cdot8^4 + 0\cdot8^3 + 1\cdot8^2 + 0\cdot8^1 - 1\cdot8^0).
\end{align*}}
\end{example}

Now we present the moves of the players. Adam has two types of moves: \textsc{regular moves} that correspond to the moves in the 2RGS and \textsc{state-check moves}, $\{\ch(i)\mid i\in \T\}$.
The moves of Eve correspond to moves in $E_1$ where incrementing and decrementing of the second counter is by $4\cdot8^n$ rather than by 1. Let $(s,(x,y),t)\in E_1$, then $\add(1,x)+\add(2,y)+\move(s,t)=(x,4y\cdot 8^n-8^s+8^t)\in E$. We call these moves \textsc{regular moves}. We also need a move for Eve to finish the simulation by removing any values corresponding to the automaton if the state is $s_{00}$. That is, we add moves $\{\move(s_{00},\top_{00})-\alpha\mid \alpha\in A_1\}$. The other type of moves, \textsc{state-defence moves}, are used to empty the counters. As in the previous construction, Eve will be able to cancel every Adam's move and decrement the counters at the same time.

Finally, we define moves connecting the simulating partition of Eve's automaton with the emptying partition. For each state $s_{ab}\in Q$ where $a,b$ are not both zero, Eve has a move $\{\move(s_{ab},k)-\ch(i) \mid (a,b)\in\{0,+\}^2\setminus\{(0,0)\},k\in\{\top_{ab},\top'_{ab}\},k\neq i,i\in\T\}$. For $s_{00}$, Eve has a move $\{\move(s_{00},\top_{00})-\ch(i)\mid i\in \T\}$.

\newlength{\myl}
\settowidth{\myl}{\small$\ch(i)$}
\newlength{\myla}
\settowidth{\myla}{Adam's}
\newlength{\mylb}
\settowidth{\mylb}{\small$\{\add(1,-4)+\add(2,-1)+\move(j,1)-\ch(i)\mid j\in\{\top_{++},\top'_{++}\}\}$}

{\small
\begin{center}
\begin{tabular}{|p{\myl}|c|}
\hline
\centering Adam's move & Eve's move \\
\hline
\multirow{5}{*}{$\alpha\in A_1$} & $\{\add(1,-4)+\add(2,-1)-\move(j,k)-\alpha\mid j,k\in\{\top_{++},\top'_{++}\},j\neq k\}$ \\
&$\{\add(1,-4)-\move(j,k)-\alpha\mid j,k\in\{\top_{+0},\top'_{+0}\},j\neq k\}$ \\
&$\{\add(2,-1)-\move(j,k)-\alpha\mid j,k\in\{\top_{0+},\top'_{0+}\},j\neq k\}$ \\
\cline{2-2}
& $\{\add(1,-4)+\add(2,-1)+\move(j,k)-\alpha\mid j\in\{\top_{++},\top_{++}'\},k\in\T, j\neq k\}$\\
\cline{2-2}
& $\{\add(1,-4)+\add(2,-1)+\move(j,1)-\alpha\mid j\in\{\top_{++},\top'_{++}\}\}$ \\
\cline{2-2}
& $\{\add(1,-4)+\move(j,1)-\alpha\mid j\in\{\top_{+0},\top'_{+0}\}\}$ \\
\cline{2-2}
& $\{\add(2,-1)+\move(j,1)-\alpha\mid j\in\{\top_{0+},\top'_{0+}\}\} $ \\
\hline
\hline
\multirow{5}{*}{$\ch(i)$} & $\{\add((1,-4e_1)+\add(2,-e_2)-\ch(i)\mid e_1,e_2\in\{0,1\}\}$ \\
\cline{2-2}
& \parbox[t]{\mylb}{$\{\add(1,-4)+\add(2,1)+\move(j,k)-\ch(i)\mid i,j\neq k, \\ \phantom{emp} j\in\{\top_{++},\top_{++}'\}, k\in\T\}$} \\
\cline{2-2}
& $\{\add(1,-4)+\add(2,-1)+\move(j,1)-\ch(i)\mid j\in\{\top_{++},\top'_{++}\}\}$ \\
\cline{2-2}
& $\{\add(1,-4)+\move(j,1)-\ch(i)\mid j\in\{\top_{+0},\top'_{+0}\}\}$ \\
\cline{2-2}
& $\{\add(2,-1)+\move(j,1)-\ch(i)\mid j\in\{\top_{0+},\top'_{0+}\}\} $ \\
\hline
\end{tabular}
\end{center}}

Next we prove which player has a winning strategy in the scenarios presented previously.

\begin{lemma}\label{2rgwin}
If both players only play \textsc{regular moves} and Eve plays only correct \textsc{regular moves}, then Eve has a winning strategy if and only if she has a winning strategy in two-dimensional robot games with states.
\end{lemma}
\begin{proof}
It easy to see that \textsc{regular moves} of the players simulate the 2RGS and that Eve has a winning strategy to reach a configuration $(s_{00},(0,0))$ of the 2RGS if and only if she has a winning strategy to reach the vector $(0,0\cdot4\cdot8^n+8^{s_{00}}-8^{\top_{00}})$ in 2RG after which Eve wins by playing $\move(s_{00},\top_{00})-\alpha$, where $\alpha$ is the \textsc{regular move} played by Adam.
\end{proof}

\begin{lemma}\label{2rgevecheats}
If Eve plays an incorrect move, i.e., after her turn the coefficient of some $8^s$ is $-1$ or the coefficient of $8^{\top_{00}}$ is zero, Adam has a winning strategy starting with a \textsc{state-check}.
\end{lemma}
\begin{proof}
First, we prove that Eve loses if a coefficient corresponding to a state of 2RGS is negative after one of her turns. A coefficient corresponding to a state of 2RGS can only be increased, namely incremented, by Eve's \textsc{regular moves}. Hence, if one of the coefficients becomes negative, then Adam wins by playing a \textsc{state-check} move. The reasoning is now similar to the usage of the \textsc{positivity check} in Lemma~\ref{2rgsevecheats}. We consider the second counter modulo $4\cdot8^n$. Before Adam's \textsc{state-check}, the configuration is in $[0,8^n) \pmod {4\cdot8^n}$ and after the check in $[3\cdot8^n,4\cdot8^n)\pmod{ 4\cdot8^n}$. If Eve does not play a \textsc{state-defence move} (a move containing a $\ch(i)$), then Adam has a winning strategy by playing a \textsc{state-check} if the second counter is not in $[3\cdot8^n,4\cdot8^n)\pmod{ 4\cdot8^n}$ and a \textsc{regular move} otherwise (recall that Adam's \textsc{regular moves} do not modify the second counter). Thus Eve has to play a \textsc{state-defence move} which does not make the negative coefficient non-negative. Now at least one of the coefficients in $\T$ is non-zero, say $i$. Adam will play $\ch(i)$ forcing Eve to play a move containing $-\ch(i)$ which will make another coefficient in $\T$ non-zero. As long as Adam keeps playing the correct \textsc{state-check}, Eve cannot make all the coefficients zero and thus cannot win.

The second case where a coefficient of some state in $\T$ is negative has been proven above. For the final case, where the coefficient of $8^{\top_{00}}$ is zero, we consider the next move of Eve. During her next turn, Eve has to play a move containing $\move(s,t)$ making the coefficient of $8^s$ negative, which has been covered previously.
\end{proof}

\begin{lemma}\label{2rgadamcheats}
If Eve plays only correct \textsc{regular moves} until Adam plays a \textsc{state-check} for the first time, then Eve has a winning strategy starting with a \textsc{state-defence move}.
\end{lemma}
\begin{proof}
Similarly as in the previous proof, if Eve does not play a \textsc{state-defence move}, then Adam has a winning strategy. Now, Eve plays the \textsc{state-defence move} $\move(s_{ab},k)-\ch(i)$ where $s_{ab}$ is the non-zero coefficient, $\ch(i)$ is the \textsc{state-defence move} Adam played and $k\in\{\top_{ab},\top'_{ab}\}$, $k\neq i$. From that point onward, Eve can empty the counters ensuring as she has emptying moves with an opposite move of Adam. Eventually, Eve will reach the configuration $(0,0)$ and win the game.
\end{proof}

\begin{lemma}\label{2rgeveresponds}
If Adam plays only \textsc{regular moves} and Eve plays a \textsc{state-defence move}, then Adam has a winning strategy starting with a \textsc{regular move}.
\end{lemma}
\begin{proof}
Since all \textsc{state-defence moves} subtract $-8^n$ from the second counter, after Eve's move, the counter is in $[8^n,2\cdot8^n)\pmod{4\cdot 8^n}$. As in proof of Lemma~\ref{2rgevecheats}, Adam ensures that the second counter does not return to the interval $[0,8^n)\pmod{4\cdot 8^n}$.
\end{proof}

\begin{proof}[Proof of Theorem~\ref{2RGSiff2RG}]
Let $(A,E)$ be the robot game constructed in this section. Assume first that Eve has a winning strategy in $(A_1,E_1)$. Now, Eve's winning strategy in two-dimensional robot games is to follow the strategy of $(A_1,E_1)$ as long as Adam plays \textsc{regular moves} which is a winning strategy by Lemma~\ref{2rgwin}. If Adam plays a \textsc{state-check}, then Eve responds according to the winning strategy of Lemma~\ref{2rgadamcheats}. 

Assume then that Adam has a winning strategy in $(A_1,E_1)$ and Eve has a winning strategy in $(A,E)$. If Adam plays only \textsc{regular moves}, then by Lemma~\ref{2rgwin}, Eve does not win by playing just the correct the correct \textsc{simulating moves}. That is, Eve has to, at some point, either play an incorrect \textsc{simulation move} or play a \textsc{state-defence move}. By Lemmas~\ref{2rgevecheats} and \ref{2rgeveresponds}, Adam has winning strategies for both cases. As we analysed all the possible moves of Eve, we have shown that Eve does not  have a winning strategy.
\end{proof}

\begin{corollary}\label{2RG}
Let $(A,E)$ be a two-dimensional robot game and an initial vector $\bx_0$. It is undecidable whether Eve has a winning strategy to reach $(0,0)$ from $\bx_0$.
\end{corollary}

\noindent Corollary~\ref{2RG} follows from Corollary~\ref{2RGS} and Theorem~\ref{2RGSiff2RG}. 
It is possible to apply it to matrix games introduced in \cite{HHNP15CiE} to show undecidability in $\Z^{3\times3}$; see the proof in the Appendix.\\

\noindent
\textbf{Final remarks: }
The construction of robot games with states was first presented in the PhD thesis of one of the authors, \cite{Reichert-phd15}, where it was also proved that robot games in dimension three are undecidable. The undecidability of 2RG is proved by a new technique of embedding 
state transitions of a 2CM into integers. It would be interesting to see whether the same approach can be applied to other automata and games. It is not clear how to embed not only state transitions of an automaton, but also the input word.

Korec showed in \cite{K96} that there exists a universal Minsky machine with 32 instructions. The natural question of a universal game arises: Is it possible to construct a fixed robot game simulating a universal 2CM? This game would have fixed moves and only the initial vector would affect the result.
In \cite{HNP15}, it was proven that two-dimensional robot games where both players have two moves are decidable in polynomial time. Consider the machine with 32 instructions. We can construct a robot game from it and count the number of moves. Thus it is undecidable whether Eve has a winning strategy in a two-dimensional robot game where Eve has at least $2083$ moves and Adam has 8 moves.
%
\setcounter{degree}{\value{theorem}}
%

\bibliographystyle{plain}
\bibliography{references}

\newpage

\appendix
\section{Appendix}
\setcounter{temp}{\value{theorem}}
\setcounter{theorem}{\value{reach}}
\begin{theorem}
Let $(Q,T)$ be a deterministic two-counter machine. It is undecidable whether in the run of $(Q,T)$, a configuration in $Q\times\{(0,0)\}$ appears.
\end{theorem}
\begin{proof}
Let $(Q,T)$ be a deterministic two-counter Minsky machine. We build $(Q',T')$ such that the run in $(Q,T)$ reaches a configuration in $\{\bot\}\times\N^2$ if, and only if, the run in $(Q',T')$ returns to a configuration in $Q\times \{(0,0)\}$. To do this, we shift the values of $(Q',T')$ to ensure that at no point both counters are zero and add emptying states used to empty the counters before reaching the halting state $\bot$. That is, a configuration $(q,(0,0))$ is reachable if and only if $q=\bot$ is reachable. Zero-checks and decrementing are implemented by shifting the counter value back down and then performing the zero-check after which the counter value is shifted back up. It follows that $(0,0)$ is reachable in $(Q',T')$ if and only if $(Q,T)$ halts.
\end{proof}

\setcounter{theorem}{\value{temp}}

\begin{corollary}
It is undecidable whether Eve has a winning strategy in a two-dimensional robot game where Eve has at least $2083$ moves and Adam has 8 moves.
\end{corollary}
\begin{proof}
We count the number of moves each player has in the robot game constructed in the previous section. Adam has 8 moves and Eve has at most $58m+227$ moves, where $m$ is the number of states in the original 2-counter machine. Korec showed in \cite{K96} that there is a universal 2-counter machine with 32 instructions, which means that it has at most 32 states. Thus there are at most 2083 moves for Eve.
\end{proof}

\noindent\textbf{Matrix games:} We apply the result on two-dimensional robot games to matrix games introduced in \cite{HHNP15CiE}.
A \emph{matrix game on vectors} (or \emph{matrix game} for short) consists of two players, Eve and Adam, having sets of
linear transformations  $\{U_1,\ldots,U_r\}\subseteq \Z^{n \times n }$ and $\{V_1,\ldots,V_s\}\subseteq\Z^{n \times n }$
respectively, an \emph{initial vector} $\bx_0 \in \Z^n$ of the game representing the starting position, and a \emph{target vector} $\by\in\Z^n$. 
Starting from $\bx_0$, players move the current point by applying available linear transformations (by matrix multiplication)
from their respective sets in turns. 
%
The decision problem of the matrix game is to check whether there exist a winning
strategy for Eve to reach the target from the starting point (vectors in $\Z^n$) of the game. 
Note that in our formulation the vectors are horizontal and players multiply it from the right.

In \cite{HHNP15CiE}, it was proven that the game is undecidable starting from dimension four. By encoding $2$-dimensional robot game into matrices, we get a matrix game of dimension three for which it is undecidable whether Eve has a winning strategy. 

\begin{theorem}
It is undecidable whether Eve has a winning strategy in three-dimensional matrix games.
\end{theorem}
\begin{proof}
Let $(A,E)$ be a two-dimensional robot game with initial vector $(x_0,y_0)$. For each move $(x,y)$ of Eve (Adam) in robot game, Eve (Adam) has a respective move
$\begin{psmallmatrix}
1 & 0 & 0 \\
x & 1 & y \\
0 & 0 & 1
\end{psmallmatrix}.$ The initial vector of matrix game is $(x_0,1,y_0)$ and the target is $(0,1,0)$. It is easy to see that adding a vector $(x,y)$ to a vector $(u,v)$ in robot games corresponds to matrix multiplication
\begin{align*}
(u,1,v)\begin{pmatrix}
1 & 0 & 0 \\
x & 1 & y \\
0 & 0 & 1 
\end{pmatrix}=(u+x,1,y+v).
\end{align*}
It is easy to see that, in the matrix game, Eve has a winning strategy to reach $(0,1,0)$ if and only if she has a winning strategy to reach $(0,0)$ in the robot game. Since the latter problem is undecidable, so is deciding whether Eve has a winning strategy in three-dimensional matrix games.
\end{proof}

\begin{figure}[htb]\centering
\scalebox{0.65}{
\begin{tikzpicture}
  
  \node[state] 		   (s00)                    {$s_{00}$};
  \node[state]         (s+0) [below = 0.3cm of s00] {$s_{+0}$};
  \node[state] 		   (s0+) [below = 0.3cm of s+0] {$s_{0+}$};
  \node[state]         (s++) [below = 0.3cm of s0+] {$s_{++}$};
  \node[state]         (t+0) [right = 1.5cm of s+0]       {$t_{+0}$};
  \node[state]         (t++) [right = 1.5cm of s++]       {$t_{++}$};
  \path[-stealth] (s00) edge              node[sloped,above] {$c_1{\scriptstyle ++}$} (t+0)
                  (s+0) edge              node[sloped,above] {$c_1{\scriptstyle ++}$} (t+0)
                  (s0+) edge              node[sloped,above] {$c_1{\scriptstyle ++}$} (t++)
                  (s++) edge              node[sloped,above] {$c_1{\scriptstyle ++}$} (t++);
\end{tikzpicture}
\qquad
\begin{tikzpicture}  
  \node[state] 		   (s00)                    {$s_{00}$};
  \node[state] 		   (s0+) [below = 0.3cm of s00]                   {$s_{0+}$};
  \node[state]         (t00) [right =1.5cm of s00]       {$t_{00}$};
  \node[state]         (t0+) [right =1.5cm of s0+]       {$t_{0+}$};

  \path[-stealth] (s00) edge              node[sloped,above] {$c_1{\scriptstyle ==}0$} (t00)
  (s0+) edge              node[sloped,above] {$c_1{\scriptstyle ==}0$} (t0+);
\end{tikzpicture}
\qquad
\begin{tikzpicture}
  
  \node[state]         (t00)        {$t_{00}$};
  \node[state]         (t+0) [below = 0.3cm of t00]       {$t_{+0}$};
  \node[state]         (t0+) [below = 0.3cm of t+0]       {$t_{0+}$};
  \node[state]         (t++) [below = 0.3cm of t0+]       {$t_{++}$};
  \node[state]         (s+0) [left = 1.5cm of t+0] {$s_{+0}$};
  \node[state]         (s++) [left = 1.5cm of t++] {$s_{++}$};

  \path[-stealth] (s+0) edge              node[sloped,above] {$c_1{\scriptstyle --}$} (t+0)
                  (s++) edge              node[sloped,above] {$c_1{\scriptstyle --}$} (t++)
                  (s+0) edge              node[sloped,above] {$c_1{\scriptstyle --}$} (t00)
                  (s++) edge              node[sloped,above] {$c_1{\scriptstyle --}$} (t0+);
\end{tikzpicture}}
\caption{\label{statettrans}Transitions modifying the first counter in the modified 2CM}
\end{figure}
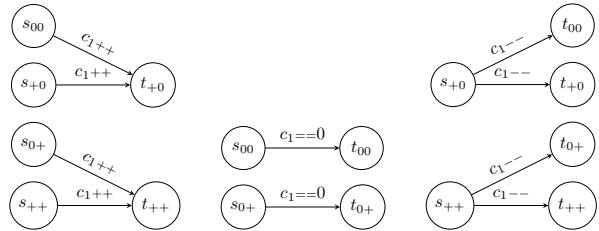

\end{document}